\documentclass[12pt]{article}

\usepackage{geometry}
\usepackage{fancyhdr}
\usepackage{amsmath}
\numberwithin{equation}{section}
\usepackage{amssymb}
\usepackage{natbib}
\usepackage{indentfirst}
\usepackage{graphicx}
\usepackage{pgfplots}
\pgfplotsset{width=0.6\textwidth,compat=1.13}
\usepackage{algorithm}
\usepackage{algorithmic}
\usepackage{epsfig}
\usepackage{chngcntr}
\usepackage{color,xcolor}
\usepackage{ntheorem}
\usepackage[colorlinks,
linkcolor=red,
anchorcolor=red,
citecolor=red]{hyperref}
\usepackage{authblk}
\usepackage{pifont}
\theorembodyfont{\upshape} 

\newcommand{\bX}{\utwi{X}}
\newcommand{\by}{\utwi{y}}
\newcommand{\ba}{\utwi{a}}
\newcommand{\bb}{\utwi{b}}
\newcommand{\utwi}[1]{\pmb{#1}} %{\mbox{\boldmath $ #1$}}
\newcommand{\balpha}{\utwi{\alpha}}

\newcommand{\bbeta}{\utwi{\beta}}

\newtheorem{Theorem}{Theorem}[section]

\newtheorem{Lemma}{Lemma}[section]

\newenvironment{proof}{{\noindent\it Proof.}}{\hfill $\square$\par}

\title{On Gibbs Sampling for Structured Bayesian Models \\
\Large Discussion of paper by Zanella and Roberts}
\author[$\dagger$]{Xiaodong Yang}
\affil[$\dagger$]{School of Gifted Young, University of Science and Technology of China}
\author[$\ddagger$]{Jun S. Liu \thanks{This paper was done while Xiaodong Yang was working as a remote summer undergraduate at Professor Jun S. Liu's lab at Harvard, in the summer of 2021.}}
\affil[$\ddagger$]{Department of Statistics, Harvard University}
\date{November 2021}
\begin{document}
\maketitle

\section{Introduction}
We congratulate Professors Giacomo Zanella and Gareth Roberts for their path-breaking work in analyzing Gibbs sampling algorithms for a class of highly practical Bayesian hierarchical models. Together with their previous work, \cite{papaspiliopoulos2003non} and \cite{papaspiliopoulos2020scalable}, their multigrid decomposition strategy elegantly reduces a high-dimensional Gibbs sampling algorithm  to independent low-dimensional components so that the convergence rate of the Gibbs sampler can be determined analytically. These are extremely interesting and encouraging results. Throughout of the article, we will  refer to this work of \cite{zanella2021multilevel} as ``Z\&R'' for simplicity. 

%\subsection{Intuitions behind multigrid decomposition}
The \textit{multigrid decomposition}  serves a central role in the whole theory established in the aforementioned series of papers. An intuition behind this decomposition is that lower-level mean statistics are sufficient for posterior inference on upper-level parameters, with lower-level parameters practically marginalized out. For example, \cite{papaspiliopoulos2003non} show that, for model \eqref{simplest} below, the posterior distribution of $(\mu,\bar{a})$ is independent of that of $(a_1-\bar{a},\cdots,a_I-\bar{a})$.

At the first glance, we cannot help notice that the intuition behind Z\&R's multigrid decomposition is quite different from that of either the classical deterministic multigrid methods \citep{mccormick1987multigrid}  or  \textit{multigrid Monte Carlo} methods \citep{goodman1989multigrid,liu2000generalised}. These latter multigrid strategies, as originally motivated by the design of efficient numerical  partial differential equation (PDE) solvers, are typically constructed artificially to accelerate the convergence of the algorithms by iterating between finer-grid and coarser-grid updates. 
%The multigrid decomposition of Z\&R is very interesting, yet it is fundamentally significant difference between the classical multigrid construction and the authors' multigrid decomposition.   
In contrast,  Z\&R's multigrid decomposition  %and classical multigrid methods differ fundamentally in that former 
is a decomposition of the given parameter space implied by the  algorithm  itself (under a specific parametrization).
%, whereas  the latter is often constructed artificially with additional auxiliary variables (e.g., parameter expansions) for accelerating convergence of the algorithm.  
Furthermore, Z\&R show that Gibbs sampling for the upper level of their multigrid decomposition converges slower than that for the lower level (Theorem 11), whereas in classical multigrid methods the upper levels are so constructed that their associated MCMC samplers converge faster  than those of the lower levels \citep{goodman1989multigrid,liu2000generalised}.

%\subsection{Multilevel Linear models}\label{sec:model}
%The multilevel linear models under their considerations could be classified as the \textit{nested} structure (majority of Z\&R) or the \textit{crossed} structure (\cite{papaspiliopoulos2020scalable} and section 4 in Z\&R). The following symmetric $3$-level hierarchical linear model with non-centering parametrization is a typical example of \textit{nested} structures: $y_{ijk},a_i,b_{ij}\in\mathbb{R}$ and $1\leq i,j,k\leq I,J,K$,
%\begin{equation}\label{nested}
%    y_{ijk} = \mu + a_i + b_{ij} + \epsilon_{ijk}.
%\end{equation}
%In contrast, model~\eqref{crossed} is a  symmetric $3$-factors crossed-effects model: %$y_{ijk},a_i,b_j,c_k,\epsilon_{ijk}\in\mathbb{R}$ and $1\leq i,j,k\leq I,J,K$,
%\begin{equation}\label{crossed}
%    y_{ijk} = \mu + a_i + b_j + c_k + \epsilon_{ijk}, \ \ 1\leq i,j,k\leq I,J,K.
%\end{equation}
%Due to space limitations and for brevity, we hope to give insights on the potentials of their multigrid decomposition and always start with the simplest model: $y_{ijk},a_i,\epsilon_{ijk}\in\mathbb{R}$ and $1\leq i,j\leq I,J$,
Despite these fundamental differences between the multigrid decomposition and  multigrid Monte Carlo, we are very much inspired by Z\&R's insightful formulation and will discuss some potential extensions of their work in the rest of the article. To illustrate our main ideas, we start by focusing on the simplest model:
\begin{equation}\label{simplest}
    y_{ij} = \mu + a_i + \epsilon_{ij}, \quad i\in[1:I],  j\in [1:J],
\end{equation}
which can be seen as either a two-level hierarchical model or a one-factor crossed-effects model. 
In the rest of the article, we use notation $\Vec{a}$ to represent a vector. For example,
$\Vec{a}_i$ used in Section 2 is an $\ell$-dimensional vector. 
Boldface letters are used  to represent collections of effects. For example, we write $\ba =(a_1,\cdots,a_I)$, and $\bar{a}$ for its mean.
We also denote $\mathbf{1}_k=(1,\cdots,1)^\top \in\mathbb{R}^{k\times1}$ and $\mathbb{I}_k$ for $k\times k$ identity matrix. For a matrix $M$, $\Vert M\Vert_2=\sqrt{\sigma_{\max}(M^\top  M)}$ denotes its spectral norm.

\section{Vector Hierarchical Models}\label{sec:vector-HM}
Our main goal here is to extend the framework of \eqref{simplest} to consider the vector-version of the model, as shown in \eqref{non-centering}. This type of models is not uncommon in practice and is a prototype of more complex realistic models. For example, the observed vector  $\Vec{y}_{ij}$ may represent several types of medical measurements (e.g., blood pressure, cholesterol level, weight, height, etc) of  individual $j$ in group $i$, and these measurements are certainly correlated within each individual. After presenting results for (\ref{non-centering}), we will comment on its potential extensions.

\subsection{Non-centering model and convergence rate}
Let us begin with an extension of model \eqref{simplest} by replacing the scalars with vectors to arrive at the following model.
\vspace{2mm}

\noindent
{\bf Model S2m} (Symmetric two-level model with \textbf{non-centering} parametrization). \textit{Suppose
\begin{equation}\label{non-centering}
    \Vec{y}_{ij}=\Vec{\mu}+\Vec{a}_i+\Vec{\epsilon}_{ij}, \ \ i\in[1:I],  j\in [1:J],
\end{equation}
where $\Vec{y}_{ij},\Vec{\mu},\Vec{a}_i,\Vec{\epsilon}_{ij}\in\mathbb{R}^\ell$, and $\Vec{\epsilon}_{ij}\stackrel{i.i.d.}{\sim} \mathcal{N}(0,\Sigma_e)$ (i.e., i.i.d. multivariate Gaussian). We impose a flat prior on $\Vec{\mu}$ and another multivariate Gaussian $\mathcal{N}(0,\Sigma_a)$ on each $\Vec{a}$. Here $\Sigma_e$ and $\Sigma_a$ are two positive definite $\ell\times \ell$ matrices.} 

\vspace{2mm}

For this model, we can write down the joint posterior distribution as
\begin{equation}
    p(\Vec{\mu},\Vec{\ba}\mid\Vec{y})\propto\exp\left[ -\frac{1}{2}\sum_{i,j}(\Vec{y}_{ij}-\Vec{\mu}-\Vec{a}_i)^\top \Sigma_e^{-1}(\Vec{y}_{ij}-\Vec{\mu}-\Vec{a}_i)-\frac{1}{2}\sum_{i}a_i^\top \Sigma_a^{-1}a_i \right].
\end{equation}
A standard Gibbs Sampler to sample from the posterior distribution $p(\Vec{\mu},\Vec{\ba}\mid\Vec{\by})$ is defined as follows.

\vspace{2mm}
\noindent
{\bf Sampler GS(0)} \textit{Initialize $\Vec{\mu}(0)$ and $\Vec{\ba}(0)$ and then iterate\\
1. Sample $\Vec{\mu}(s+1)$ from $p(\Vec{\mu}\mid \Vec{\ba}(s),\Vec{\by})$;\\
2. Sample $\Vec{a}_i(s+1)$ from $p(\Vec{a}_i\mid \Vec{\mu}(s+1),\Vec{\by})$ for  $i=1,\ldots, I$, independently.}
\vspace{2mm}

Using the same notations as in Z\&R, we define $\bar{\Vec{a}}=\sum_{i}\Vec{a}_i/I$ to be mean and
\begin{equation*}
    \delta\Vec{a}_i=\Vec{a}_i-\bar{\Vec{a}},\quad\delta\Vec{\ba}=(\delta\Vec{a}_1,\cdots,\delta\Vec{a}_I)
\end{equation*}
as the residual. Given this notation, we derive the following factorization
\begin{equation}
    p(\Vec{\mu},\Vec{\ba}\mid\Vec{\by})=p(\Vec{\mu},\bar{\Vec{a}}\mid\Vec{\by})\times p(\delta\Vec{\ba}\mid\Vec{\by}).
\end{equation}
This factorization paves the way for the following multigrid decomposition.

Before stating and proving our result, we introduce a lemma without proof to compute the $L^2$ convergence rate of some two-component Gaussian Gibbs sampler.

\begin{Lemma}\label{maximal correlation}
Let the target distribution $\pi(q_1,q_2)$, where $q_1,q_2\in\mathbb{R}^\ell$, be a $2\ell$-dimensional Gaussian distribution with $\text{var}(q_1)=\Sigma_{11}$,
\iffalse
the covariance matrix
%\begin{equation*}
    \left(\begin{matrix}
    \Sigma_{11} & \Sigma_{12}\\
    \Sigma_{21} & \Sigma_{22}
    \end{matrix}\right),
\end{equation*}
\fi
$\text{var}(q_2)=\Sigma_{22}$, and $\text{cov}(q_1,q_2)=\Sigma_{12}$. %\in\mathbb{R}^{\ell\times\ell}$ are the covariance matrices of the two components, $q_1$ and $q_2$, respectively. 
The convergence rate of the  Gibbs sampler that iterates between conditional sampling   $[q_1 \mid q_2]$ and  $[q_2\mid q_1]$  is equal to the squared spectral norm 
$\left\Vert\Sigma_{11}^{-1/2}\Sigma_{12}\Sigma_{22}^{-1/2}\right\Vert^2_2.$
\end{Lemma}

\noindent
{\it Remark.} This lemma is an easy consequence of Theorem 1 in \cite{roberts1997updating}, in which the generated Markov chain is recognized as a multivariate AR(1) process. 
See also Section 5.1, \cite{liu1994covariance}, for an elementary proof based on \textit{maximal correlations}, as this quantity can also be interpreted as the \textit{maximal correlation} between $q_1$ and $q_2$.
\vspace{2mm}

\begin{Theorem}\label{thm:non-centering}
Let $\{\Vec{\mu}(t),\Vec{\ba}(t)\}$ be the Markov chain generated by either the standard Gibbs sampler. Then the the functionals $\{\delta\Vec{\ba}(t)\}$ and $\{\Vec{\mu}(t),\bar{\Vec{a}}(t)\}$ evolve as two independent Markov chains. Furthermore, the $L^2$-convergence rate of the sampler is
\begin{equation}\label{non-centering rate}
    \rho_0=\left\Vert\left(J\Sigma_e^{-1}\right)^{1/2}\left(\Sigma_a^{-1}+J\Sigma_e^{-1}\right)^{-1/2} \right\Vert_{2}^2.
\end{equation}
\end{Theorem}

\begin{proof}
The decomposition directly follows from the following two identities
\begin{align}
    p\left[\Vec{\mu}(s+1)\mid \Vec{\ba}(s),\Vec{\by}\right]&=p\left[\Vec{\mu}(s+1)|\Vec{\bar{a}}(s),\Vec{\by}\right],\\
    p\left[\Vec{\bar{a}}(s+1),\delta \Vec{\ba}(s+1)\mid \Vec{\mu}(s+1),\Vec{\by}\right]&=p\left[\Vec{\bar{a}}(s+1)|\Vec{\mu}(s),\Vec{\by}\right]\times p\left[\delta\Vec{\ba}(s+1)\mid \Vec{y}\right].
\end{align}
Moreover, the latter identity further implies that $\{\delta\Vec{\ba}(t)\}$ carries out exact sampling. So the convergence rate of $\{\Vec{\mu}(t),\Vec{\ba}(t)\}$ is actually determined by the rate of $\{\Vec{\mu}(t),\bar{\Vec{a}}(t)\}$. The latter chain converges to the following joint-normal stationary distribution
\begin{equation*}
    p(\Vec{\mu},\bar{\Vec{a}}\mid\Vec{y})\propto\exp\left[-\frac{IJ}{2}\Vec{\mu}^\top \Sigma_e^{-1}\Vec{\mu}-\frac{1}{2}\bar{\Vec{a}}^\top \left(I\Sigma_a^{-1}+IJ\Sigma_e^{-1}\right)\bar{\Vec{a}}-IJ\Vec{\mu}^\top \Sigma_e^{-1}\bar{\Vec{a}}+IJ\bar{\Vec{y}}^\top \Sigma_e^{-1}(\Vec{\mu}+\bar{\Vec{a}})\right],
\end{equation*}
where we write $\bar{\Vec{y}}\triangleq\sum_{i,j}\Vec{y}_{ij}/IJ$. This is a Markov chain in a $2\ell$-dimensional space induced by the block-wise two-component Gibbs sampler. In contrast, the original chain is of dimension $(I+1)\ell$. The final result then follows from Lemma \ref{maximal correlation}.
\end{proof}

\vspace{3mm}

\noindent
{\it Remark.} If we choose dimension $\ell=1$ and replace $\Sigma_e$ and $\Sigma_a$ with $\sigma_e^2$ and $\sigma_a^2$, respectively, the convergence rate becomes
$$
\rho_0=\dfrac{J\sigma_e^{-2}}{\sigma_a^{-2}+J\sigma_e^{-2}},
$$
which coincides with Proposition 3 in \cite{papaspiliopoulos2020scalable}.

\subsection{Convergence rate for centering model}
Inspired by Z\&R, we seek to give a theoretical guidance towards centering \eqref{non-centering} or non-centering \eqref{centering} parametrizations.

\vspace{2mm}
\noindent
{\bf Model S2m} (Symmetric two-level model with \textbf{centering} parametrization). \textit{Suppose
\begin{equation}\label{centering}
    \Vec{y}_{ij}\sim\mathcal{N}(\Vec{\alpha}_i,\Sigma_e),\quad\Vec{\alpha}_i\sim\mathcal{N}(\Vec{\mu},\Sigma_a), \ i\in[1:I],  j\in [1:J],
\end{equation}
where $\Vec{y}_{ij},\Vec{\mu},\Vec{\alpha}_i %,\Vec{\epsilon}_{ij}
\in\mathbb{R}^\ell$. Same as before, a flat prior is imposed on $\Vec{\mu}$. Here $\Sigma_e$ and $\Sigma_a$ are two positive definite $\ell\times \ell$ matrices.}
\vspace{2mm}

\noindent
{\bf Sampler GS(1)} \textit{Initialize $\Vec{\mu}(0)$ and $\Vec{\balpha}(0)$ and then iterate\\
1. Sample $\Vec{\mu}(s+1)$ from $p(\Vec{\mu}\mid \Vec{\balpha}(s),\Vec{\by})$;\\
2. Sample $\Vec{\alpha}_i(s+1)$ from $p(\Vec{\alpha}_i\mid \Vec{\mu}(s+1),\Vec{\by})$ for   $i=1,\ldots, I$ independently.}
\vspace{2mm}

Almost in the same manner, we offer the following theorem.
\begin{Theorem}\label{thm:centering}
Let $\{\Vec{\mu}(t),\Vec{\balpha}(t)\}$ be the Markov chain generated by the  sampler GS(1). Then the  functionals $\{\delta\Vec{\balpha}(t)\}$ and $\{\Vec{\mu}(t),\bar{\Vec{\alpha}}(t)\}$ evolve as two independent Markov chains. Furthermore, the $L^2$-convergence rate of $\{\Vec{\mu}(t),\Vec{\balpha}(t)\}$ is
\begin{equation}\label{centering rate}
    \rho_1=\left\Vert \left(\Sigma_a^{-1}\right)^{1/2}\left(\Sigma_a^{-1}+J\Sigma_e^{-1}\right)^{-1/2} \right\Vert_{2}^2.
\end{equation}
\end{Theorem}

\vspace{3mm}

\noindent
{\bf Optimal Parameterization Strategy:} If $\rho_0\leq\rho_1$, then choose the non-centering parameterization \eqref{non-centering}; otherwise, choose the centering parameterization \eqref{centering}.
\vspace{2mm}

When dimension $\ell=1$, \eqref{centering rate} becomes 
$\rho_1=\sigma_a^{-2}/(\sigma_a^{-2}+J\sigma_e^{-2})$. This strategy can be adaptively used when the variances are unknown. Specifically, in one iteration, after sampling $\hat{\sigma}_a^2,\hat{\sigma}_e^2$, we  compare $J\hat{\sigma}_e^{-2}/(\hat{\sigma}_a^{-2}+J\hat{\sigma}_e^{-2})$ and $\hat{\sigma}_a^{-2}/(\hat{\sigma}_a^{-2}+J\hat{\sigma}_e^{-2})$, and choose the optimal parameterization accordingly. A direct benefit is that we can always achieve a convergence rate bounded by $1/2$ since $\rho_0+\rho_1=1$, regardless of what values $\sigma_a^2,\sigma_e^2$ are \citep{papaspiliopoulos2003non}. Corollary 
2 in Z\&R proposes an optimal parametrization strategy for  $3$-level models and gives a constant rate upper bound $2/3$ therein.

However, in a multi-dimensional case with $\ell>1$, the rates found in Theorem \ref{thm:non-centering} and Theorem \ref{thm:centering} do not necessarily sum up to $1$. Though the parameterization strategy still applies, it does not necessarily give a constant rate upper bound. If both covariance matrices are diagonal, i.e., $\Sigma_a=\text{diag}(1/\tau^a_1,\cdots,1/\tau^a_\ell)$ and  $\Sigma_e=\text{diag}(1/\tau^e_1,\cdots,1/\tau^e_\ell)$, then we have
\begin{equation*}
    \rho_0=\max_{1\leq i\leq \ell}\left[\frac{J\tau^e_i}{\tau^a_i+J\tau^e_i}\right], \ \ \ 
    \rho_1=\max_{1\leq i\leq \ell}\left[\frac{\tau^a_i}{\tau^a_i+J\tau^e_i}\right].
\end{equation*}
Applying the optimal paramterization strategy component-wise is of interest in this non-correlated case. That is, we may introduce a ``centering" indicator variable $C$ of dimension $\ell$, indicating which of the $\ell$ components use centering and which use non-centering parameterization. In this way, we may still be able to obtain the rate bound $1/2$.

When  $\Sigma_a$ and $\Sigma_e$ become general non-diagonal covariance matrices, the picture becomes  more complicated. It will be of great interest to develop some methodological guidance on how to approach this problem. The constant rate bound $1/2$ as discussed above is no longer guaranteed, and it is entirely possible that both rates are close to 1. We speculate that one may extend the ``centering" indicator $C$ to be a continuous vector to allow ``partial-centering'' (more about this issue in Section~\ref{sec:partial-centering}).

It is also not too difficult to extend these results to more complex structures such as  three-level vector hierarchical models and vector crossed-effects models, although the formulae would grow more  complicated and the design of the optimal parameterization may no longer be possible. The authors' insights and suggestions along this direction would be very much welcome.

\section{Incorporating Regression Covariates}

%Our main goal here is to raise discussions towards Bayesian linear regression, via considering 
Zanella and Roberts mainly focus on hierarchical models with certain symmetry conditions for data without individual-level covariates. 
Mixed-effects models, which accommodate individual-level variability and are very commonly used in practice, seem to have not been directly covered by Z\&R. Our goal here is to consider possible ways to extend the authors' multigrid decomposition technique to this more complex class of models.
%However, the multigrid decomposition technique developed by Professors Zanella and Roberts, only applies to restrictive cases.

\subsection{Linear mixed effects models}
To extend and see the limits of  multigrid decomposition, we consider the following simple extension, which just replaces the intercept term $\mu$ with a linear combination of $p$ covariates with a fixed coefficient vector. Previously, \cite{gao2019estimation} attempted to tackle the computational efficiency of this model \eqref{regression}. But their results give loose bounds while requiring mild conditions.

\vspace{2mm}
\noindent
{\bf Model SR} (Symmetric two-level mixed-effect model). \textit{Suppose
\begin{equation}\label{regression}
    y_{ij}=X_{ij}^\top \beta+a_i+\epsilon_{ij},\ \ i\in[1:I],  j\in [1:J],
\end{equation}
where $\epsilon_{ij}$ is i.i.d. normal random variables with mean 0 and variance $\sigma^2_e$. Moreover, $X_{ij},\beta\in\mathbb{R}^p$ (column vectors) are known covariates and unknown coefficients respectively. We then impose a standard Bayesian model specification assuming $a_i\sim\mathcal{N}(0,\sigma^2_a)$ and $\beta\sim\mathcal{N}(0,\Sigma_0)$.}
\vspace{2mm}

Essential full-rank conditions should be imposed on the design matrix. Requiring $p<I$, we denote the $I\times p$ matrix as
\begin{equation*}
\bar{X}\triangleq (\bar{X}_1,\ldots, \bar{X}_I)^\top ,
%\left(
%\begin{matrix}
%\bar{X}_1(1) & \bar{X}_1(2) &\cdots &\bar{X}_1(p)\\
%\bar{X}_2(1) & \bar{X}_2(2) &\cdots &\bar{X}_2(p)\\
%\cdots&&&\cdots\\
%\bar{X}_I(1) & \bar{X}_I(2) &\cdots &\bar{X}_I(p)
%\end{matrix}\right),    
\end{equation*}
where $\bar{X}_i=J^{-1}\sum_{j}X_{ij}\in\mathbb{R}^p$.
%Its $k$-th coordinate is written as $\bar{X}_i(k)$.
A further natural requirement is that $\bar{X}$ is of rank $p$. Then, we can  define a $p\times I$ matrix $P=(\bar{X}^\top \bar{X})^{-1/2}\bar{X}^\top $.  We also introduce another $(I-p)\times I$ matrix $L$ such that $L^\top  L +P^\top P=\mathbb{I}_I$ (i.e., the identity matrix of dimension $I$). Note that $P^\top  P =\bar{X}(\bar{X}^\top \bar{X})^{-1} \bar{X}^\top $ and $PP^\top =\mathbb{I}_p$. Let $\bX =\{X_{ij} \}$.

\vspace{2mm}
\noindent
{\bf Sampler GS (Regression)} \textit{Initialize $\beta(0)$ and $\ba(0)$ and then iterate\\
1. Sample $\beta(s+1)$ from $p(\beta\mid \ba(s),\bX,\by)$;\\
2. Sample $a_i(s+1)$ from $p(a_i\mid \ba(s+1),\bX,\by)$ for all $i$.}
\vspace{2mm}

\begin{Theorem}\label{thm:regression}
Let $\{\beta(t),\ba(t)\}$ be the Markov chain generated by the standard Gibbs sampler. Then the two functionals $\{L\ba(t)\}$ and $\{\beta(t),\bar{X}^\top \ba(t)\}$ evolve as two independent Markov chains. Furthermore, the $L^2$-convergence rate of $\{\beta(t),\ba(t)\}$ is
\begin{equation}\label{rate regression}
    \rho=\frac{J^2\sigma_e^{-4}}{\sigma_a^{-2}+J\sigma_e^{-2}}\left\Vert \left(\bar{X}^\top \bar{X}\right)^{1/2}\left(\Sigma_0^{-1}+\sum_{i,j}X_{ij} X_{ij}^\top  \sigma_e^{-2}\right)^{-1/2} \right\Vert_{2}^2.
\end{equation}
\end{Theorem}

\begin{proof}
It is easy to write down the likelihood function and prior:
\begin{align*}
    p(\by\mid \bX,\beta,\ba)&\propto \prod_{i=1}^I \prod_{j=1}^J  \exp\left[-\frac{1}{2\sigma_e^2}(y_{ij}-X_{ij}^\top  \beta-a_i)^2\right],\\
    p(\beta,\ba)&\propto \exp\left[-\frac{1}{2}\beta^\top \Sigma_0^{-1}\beta-\frac{1}{2\sigma^2_a}\sum_{i=1}^{I}a_i^2\right].
\end{align*}
The posterior distribution is
\begin{align*}
    p(\beta,\ba\mid \by,\bX)\propto&\exp\left[-\frac{1}{2}\beta^\top \Sigma_0^{-1}\beta-\frac{1}{2\sigma^2_a}\sum_{i}a_i^2-\frac{1}{2\sigma_e^2}\sum_{i,j}(y_{ij}
    -X_{ij}^\top  \beta-a_i)^2\right]\\
    \propto&\exp\left[-\frac{1}{2}\beta^\top \left(\Sigma_0^{-1}+\sum_{i,j}X_{ij} X_{ij}^\top \sigma_e^{-2}\right)\beta-\frac{1}{2}\left(\frac{1}{\sigma^2_a}+\frac{J}{\sigma^2_e}\right)\sum_i a_i^2-\frac{1}{\sigma_e^2}\sum_{i,j}a_i
    X_{ij}^\top \beta\right]\\
    &\exp\left[\dfrac{J}{\sigma^2_e}\sum_{i}a_i\bar{y}_i+\dfrac{1}{\sigma_e^2}\sum_{ij}y_{ij}
    X_{ij}^\top \beta\right].
\end{align*}
We should especially focus on the cross term
\begin{equation*}
    \sum_{ij}a_iX_{ij}^\top  \beta=
    \sum_{i=1} a_i (J \bar{X}_i^\top ) \beta
    %\sum_{k=1}^{p}\beta_k\left(\sum_{i=1}^I\bar{X}_i(k)a_i\right)
    =J\ba^\top \bar{X}\beta.
\end{equation*}
Furthermore, we also find that
\begin{equation*}
    \sum_i a_i^2=\ba^\top \ba=\Vert P\ba\Vert^2+\Vert L\ba\Vert^2=\ba^\top \bar{X}\left(\bar{X}^\top \bar{X}\right)^{-1}\bar{X}^\top \ba+\Vert L\ba\Vert^2.
\end{equation*}
The distribution of $\ba$ is actually equivalent to the joint distribution of $(\bar{X}^\top \ba,L\ba)$, since $(\bar{X},L^\top )$ is an invertible $I\times I$ matrix. Hence, we derive the following factorization
\begin{equation}
    p(\beta,\ba\mid \by,\bX)=p(\beta,\bar{X}^\top \ba\mid \by,\bX)\times p(L\ba\mid \by,\bX).
\end{equation}
We shall also deduce the following identities
\begin{align}
    p\left[\beta(s+1)\mid\ba(s),\by,\bX\right] =& p\left[\beta(s+1)\mid\bar{X}^\top \ba(s),\by,\bX\right],\\ p\left[\bar{X}^\top \ba(s+1),L\ba(s)\mid\beta(s),\by,\bX\right] =& p\left[\bar{X}^\top \ba(s+1)\mid\beta(s),\by,\bX\right]p\left[L\ba(s)\mid\by,\bX\right],
\end{align}
which imply the multigrid decomposition. Again,  convergence rate $\rho$ is controlled by the convergence rate of $\{\beta(t),\bar{X}^\top \ba(t)\}$. The joint target distribution of $\{\beta,\bar{X}^\top \ba\}$ is
\begin{align*}
    p(\beta,\bar{X}^\top \ba\mid \by,\bX)\propto&\exp\left[ -\frac{1}{2}\beta^\top \left(\Sigma_0^{-1}+\sum_{i,j}X_{ij}^\top  X_{ij}\sigma_e^{-2}\right)\beta-\frac{J}{\sigma_e^2}\ba^\top \bar{X}\beta\right]\\
    &\exp\left[ -\frac{1}{2}\left(\frac{1}{\sigma^2_a}+\frac{J}{\sigma^2_e}\right)\ba^\top \bar{X}\left(\bar{X}^\top \bar{X}\right)^{-1}\bar{X}^\top \ba\right]
\end{align*}
By Lemma \ref{maximal correlation},  the $L^2$ convergence rate is equal to the squared maximal correlation between $\beta$ and $\bar{X}^\top \ba$.
\end{proof}

\vspace{2mm}

\noindent
{\it Remark 1.} If we set $p=1,X_{ij}\equiv1$, then  $\bar{X}_i=1$, $\bar{X}^\top  \bar{X}=I$ and $\sum_{ij}X_{ij}^\top X_{ij}=IJ$. By placing a flat prior on $\mu$, we just replace $\Sigma_0^{-1}$ with $0$ in \eqref{rate regression}. Henceforth, Theorem \ref{thm:regression} reduces to $\rho=J\sigma_e^{-2}/(\sigma_a^{-2}+J\sigma_e^{-2})$, in this case.

\vspace{2mm}
\noindent
{\it Remark 2.} Theorem \ref{thm:regression} implies that summary statistics $\bar{X}^\top \ba$ of the lower level parameters are sufficient for the inference of upper level parameters $\beta$, with $L\ba$ marginalized out.

\vspace{2mm}
\noindent
{\it Remark 3.} Further note that \eqref{rate regression} is invariant if the variance terms are scaled simultaneously. Specifically, \eqref{rate regression} remains the same if we replace  $\left(\Sigma_0,\sigma_a^2,\sigma_e^2\right)$ by $\left(r\Sigma_0,r\sigma_a^2,r\sigma_e^2\right)$ where $r>0$. Moreover, another common rotation invariance in Bayesian linear regression applies to our result: \eqref{rate regression} remains the same if the pair $\left(\Sigma_0,X_{ij}\right)$   is replaced with $\left(R^\top \Sigma_0 R, RX_{ij} \right)$, where $R$ is a $p\times p$ orthogonal matrix.

\vspace{2mm}
We further note that the multigrid decomposition techniques do not naturally extend to more complex structures. Roughly speaking, both nested structures (such as $y_{ijk}=X_{ijk}^\top \beta+a_i+b_{ij}+\epsilon_{ijk}$) and crossed structures (such as $y_{ijk}=X_{ijk}^\top\beta+a_i+b_j+\epsilon_{ijk}$) would bring in a new cross term ``$\ba^\top \bb$", which is hard to handle. Can we still obtain an elegant decomposition for these models?

Indeed, many researchers have studied the general linear mixed-effects model:
\begin{equation}\label{general mixed effects}
    y=X^\top \beta+Z^\top u+\epsilon,
\end{equation}
where, in the first part, $\beta$ is common to all individuals as in a typical linear regression framework, and $u$ represents random effects (e.g., $Z$ can be dummy variables). For example, if $Z$ represents  one categorical variable with $I$ categories (using a dummy variable representation), this general form \eqref{general mixed effects} reduces to the simple model \eqref{regression} considered before.

Model \eqref{general mixed effects} with arbitrary  $Z$, however, has an identical mathematical representation as a standard linear regression model (i.e., one can simply treat $(X,Z)$ as covariates) although the prior distributions for $\beta$ and $u$ may differ substantially.  
Compared with the models handled in Z\&R, a key thing we have lost in the general model \eqref{general mixed effects} seems to  be the strong symmetry that can be used to decompose the involved variables into meaningful levels. A curious question is: how far we can push so that we can still have certain meaningful decomposition? 

\subsection{Implications  for  general linear regression models}
%The linear regression model has long been a focus of statisticians. 
\subsubsection{Linear model formulation of two-level hierarchical model}
We can recast the multigrid decomposition of Z\&R for both centering and non-centering parameterizations of model \eqref{simplest} in the context of general Bayesian linear regression via covariate orthogonalization.
%We study this orthogonalization via non-cetering or centering parametrizations of \eqref{simplest}. 

\vspace{2mm}
\noindent
{\bf Non-centering Parametrization.} By setting $\beta=(a_1,\cdots,a_I)^\top $ and
\begin{equation}\label{y,X}
y=(y_{11},y_{12},\cdots,y_{1I},y_{21},\cdots,y_{IJ})^\top \in\mathbb{R}^{IJ\times1},X=(\mathbb{I}_{I}\otimes\mathbf{1}_J)^\top\quad(\text{Kronecker product}),
\end{equation}
the simple linear model $y=\mu\mathbf{1}_{IJ}+X^\top\beta+\epsilon$ is equivalent to model \eqref{simplest}. The decomposition can be seen as imposing a linear transformation by replacing $\beta$ with $A\beta$, where the first row of $A$ is $\frac{1}{\sqrt{I}}\mathbf{1}_I^\top $ and $A$ is $I\times I$ orthogonal. With flat prior on $\mu$ and independent $\mathcal{N}(0,1/\tau_a)$ on each $a_i$, the posterior is
\begin{align*}
p(\beta,\mu\mid y,X)\propto&\exp\left(-\frac{1}{2}\beta^\top (\tau_e X X^\top+\tau_a\mathbb{I})\beta-\tau_e\mu\mathbf{1}^\top _{IJ}X^\top\beta-\frac{IJ\tau_e}{2}\mu^2\right)\\
=&\exp\left(-\frac{1}{2}(A\beta)^\top (\tau_e AX X^\top A^\top +\tau_a\mathbb{I})(A\beta)-\tau_e\mu[A\beta]_1-\frac{IJ\tau_e}{2}\mu^2\right).
\end{align*}
Moreover, $[AX X^\top A^\top ]_{i1}=[AXX^\top A^\top ]_{1i}=0$ for any $i\geq2$, which means that the first column of $X^\top A^\top $ is orthogonal to the other columns. %This identity implies the independence between
Thus, $(\mu,[A\beta]_1)$ and $[A\beta]_{2:I}$ are independent \textit{a posteriori}. The first component corresponds to $(\mu,\bar{a})$ and the latter one is a representation of the residual $\delta a$. The multigrid decomposition is then built upon this orthogonalization.  To investigate the potential of this orthogonalization-based view, we consider the following general linear regression model.

\vspace{2mm}

\noindent
{\bf Centering Parametrization} Model \eqref{simplest} can also be written as
%The centering parametrization of \eqref{simplest} 
\begin{equation}
    y_{ij}\sim\mathcal{N}(\alpha_i,1/\tau_e),\quad\alpha_i\sim\mathcal{N}(\mu,1/\tau_a), \ i\in[1:I],  j\in [1:J].
\end{equation}
Set $y$, $X$, and $\beta$ exactly the same way as \eqref{y,X}, %and set $\beta=(\alpha_1,\cdots,\alpha_I)^\top $,  
we have an equivalent model:
\begin{equation}
    y=X^\top \beta+\epsilon,\quad\beta\sim\mathcal{N}(\mu\mathbf{1}_I,1/\tau_a\mathbb{I}_{I}),\quad \epsilon\sim\mathcal{N}(0,1/\tau_e\mathbb{I}_{IJ}).
\end{equation}
Intuitively, we use a new prior on   $\beta$ with a latent variable $\mu$. With flat prior on $\mu$, the posterior is
\begin{align*}
    p(\beta,\mu\mid y,X)\propto\exp\left[-\frac{1}{2}\beta^\top \left(\tau_eXX^\top+\tau_a\mathbb{I}\right)\beta+\tau_a\mu\mathbf{1}_I^\top \beta-\frac{I\tau_a}{2}\mu^2\right].
\end{align*}
We can apply the same linear transformation $A$ as before.

\subsubsection{Extension to general linear models}
\vspace{2mm}
\noindent
{\bf Model LM} Suppose $X_1\in\mathbb{R}^{p_1\times n}$, $X_2\in\mathbb{R}^{p_2\times n}$ are two sets of covariates and consider
\begin{equation}\label{linear regression model}
    y=X_1^\top\beta_1+X_2^\top\beta_2+\epsilon,
\end{equation}
where $\beta_i\in\mathbb{R}^{p_i},(i=1,2)$ are unknown coefficients. Error $\epsilon\in\mathbb{R}^n$ is modeled as i.i.d. $\mathcal{N}(0,1/\tau_e)$. Independent priors $\mathcal{N}(0,1/\tau_1\mathbb{I}_{p_1})$ and $\mathcal{N}(0,1/\tau_2\mathbb{I}_{p_2})$ are imposed on $\beta_1$ and $\beta_2$ respectively.

\vspace{2mm}
Assume $r=\text{rank}(X_1 X_2^\top)$, we conduct SVD to find $B_i\in\mathbb{R}^{r\times p_i},(i=1,2)$ with orthonormal rows and diaginal $Q=\text{diag}(\lambda_1,\cdots,\lambda_r)$ such that
\begin{equation}
X_1 X_2^\top=B_1^\top QB_2.
\end{equation}
By constructing  orthogonal matrices $A_i\in\mathbb{R}^{p_i\times p_i}$, $i=1,2$, as  completions of $B_1$ and $B_2$, respectively, i.e., $A_i$ and $B_i$ share the same $r$ first rows, we have the following result.

\begin{Theorem}
Consider a Markov chain $\{\beta_1(s),\beta_2(s)\}$ generated by a systematic Gibbs sampler alternating between conditional sampling  $[\beta_1 \mid \beta_2]$ and $[\beta_2\mid\beta_1]$. Define $\theta_i=\left(\theta_i^{(1)},\cdots,\theta_i^{(p_i)}\right)^\top =A_1\beta_i$. Then, the evolution of $\{\theta_1(s),\theta_2(s)\}$ is equivalent to that of $\{\beta_1(s),\beta_2(s)\}$. If the first $r$ columns of $X_i^\top A_i^\top $ are orthogonal to the rest $p_i-r$ columns
\begin{equation}\label{special condition on X}
    [X_i^\top A_i^\top ]_{1:n,k_1}\perp[X_i^\top A_i^\top ]_{1:n,k_2}, \quad\forall k_1\leq r<k_2, 
\end{equation}
the evolutions of $\{\theta_{1}^{(1:r)}(s),\theta_{2}^{(1:r)}(s)\}$, $\{\theta_1^{((r+1):p_1)}\}$ and $\{\theta_2^{((r+1):p_2)}\}$ are independent.
\end{Theorem}

\begin{proof}
We start by writing out the joint posterior
\begin{align}\label{eq:post_LM}
    p(\beta\mid y,X)&\propto\exp\left[-\tau_e\beta_1^\top X_1 X_2^\top\beta_2-\frac{1}{2}\sum_{i=1}^{2}\beta_i^\top \left(\tau_e X_i X_i^\top+\tau_i\mathbb{I}_{p_i}\right)\beta_i\right]\\
    &=\exp\left[-\tau_e\left(\theta_1^{(1:r)}\right)^\top  Q\theta_2^{(1:r)}-\frac{1}{2}\sum_{i=1}^{2}\theta_i^\top \left(\tau_e A_i^\top X_i X_i^\top A_i+\tau_i\mathbb{I}_{p_i}\right)\theta_i\right]\\
    &=p\left(\theta_{1}^{(1:r)},\theta_{2}^{(1:r)}\mid y,X\right)\prod_{i=1}^2 p\left(\theta_i^{((r+1):p_i)}\mid y,X\right),
\end{align}
where the last equality follows from the condition \eqref{special condition on X}. Based on these identities,
\begin{align*}
    p\left(\theta_1^{((r+1):p_1)}(s+1)\mid y,X,\theta_2(s)\right)&=p\left(\theta_1^{((r+1):p_1)}(s+1)\mid y,X\right),\\
    p\left(\theta_2^{((r+1):p_2)}(s+1)\mid y,X,\theta_1(s+1)\right)&=p\left(\theta_2^{((r+1):p_2)}(s+1)\mid y,X\right),\\
    p\left(\theta_1^{(1:r)}(s+1)\mid y,X,\theta_2(s)\right)&=p\left(\theta_1^{(1:r)}(s+1)\mid y,X,\theta_2^{(1:r)}(s)\right),\\
    p\left(\theta_2^{(1:r)}(s+1)\mid y,X,\theta_1(s+1)\right)&=p\left(\theta_2^{(1:r)}(s+1)\mid y,X,\theta_1^{(1:r)}(s+1)\right),
\end{align*}
the conclusion of the theorem is thus proved.
\end{proof}

One implication of the result is that the multigrid decomposition developed for \eqref{simplest} is non-trivial in the sense that condition \eqref{special condition on X} must be imposed on the covariate matrix. Recall that we have written out the dummy variables $X$ explicitly  for \eqref{simplest}, and thus verified this condition implicitly for the linear model form of \eqref{simplest}.
%But we must point out that 

\vspace{3mm}

\noindent
{\bf Centering for linear models.}
Model \eqref{linear regression model} with its priors can be rewritten as
\begin{equation}
    y=X_2^\top\beta_2+\epsilon,\quad \beta_2\sim\mathcal{N}(M\beta_1,1/\tau_2\mathbb{I}_{p_2}),\quad \beta_1\sim\mathcal{N}(0,1/\tau_1\mathbb{I}_{p_1}),\quad\epsilon\sim\mathcal{N}(0,1/\tau_e\mathbb{I}_{n}),
\end{equation}
to mimic the centering parametrization, where $M\in\mathbb{R}^{p_2\times p_1}$ such that $X_1^\top=X_2^\top M$\footnote{For the simplest model \eqref{simplest}, we actually use $M=\mathbf{1}_I$.}, assuming that $M$ exists.

Now the posterior distribution is
\begin{equation}\label{eq:LM_centering_posterior}
    p(\beta\mid y,X)\propto\exp\left[ \tau_2\beta_1^\top M^\top \beta_2-\frac{1}{2}\beta_2^\top \left(\tau_eX_2 X_2^\top+\tau_2\mathbb{I}_{p_2}\right)\beta_2-\frac{1}{2}\beta_1^\top \left(\tau_2M^\top M+\tau_1\mathbb{I}_{p_1}\right)\beta_1 \right].
\end{equation}
Let the SVD of $M$ be
\begin{equation}
    M=B_1^\top QB_2,
\end{equation}
where $Q\in\mathbb{R}^r,r=\text{rank}(M)$. Again we denote the complement of $B_i$ as $A_i$. Then we require the following condition
\begin{equation}\label{eq:centering_condition}
    [X_2^\top A_2^\top ]_{1:n,k_1}\perp[X_2^\top A_2^\top ]_{1:n,k_2},\quad\forall k_1\leq r<k_2.
\end{equation}
to validate a similar multigrid decomposition. Again, this condition automatically holds for the two-level hierarchical model, but do not hold in general. 

\subsection{Thoughts and speculations}
In both the  non-centering and centering formulations, 
conditions \eqref{special condition on X} and \eqref{eq:centering_condition} most likely do not hold for an arbitrary design matrix $X$. Thus, a multigrid decomposition similar to that of Z\&R seems difficult to come by.
%mposed seems to only hold for a set of measure zero, in the usual context of i.i.d. pairs of response and covariate. 
Some natural questions arise: Does a useful multigrid decomposition exist for a general linear regression model in some other ways? If so, what would be a correct construction?
If not, how can we gain more insights on the Gibbs sampler for a general Bayesian regression model \eqref{linear regression model}? 
Can we find a good matrix $M$ so that the convergence rate of the Gibbs sampler corresponding to \eqref{eq:LM_centering_posterior} is faster than that based on \eqref{eq:post_LM}? What if the Gibbs sampler has more than two components?

Besides the  Gaussian prior we have studied here, many other prior distributions have been proposed to accommodate  both sparsity and biases in coefficient estimations, including   \textit{spike-and-slab} priors \citep{mitchell1988bayesian}, \textit{horseshoe} priors \citep{carvalho2010horseshoe}, \textit{neuronized} priors \citep{shin2021neuronized}, and so on. Can one   extend Z\&R's and our results  to accommodate other priors that are more appropriate for high-dimensional  problems? The Gaussian spike-and-slab prior  may be a most likely solvable case?

%\noindent {\bf Some Summary:}  Actually, I come up with this setting (Model LM), in order to think about what structures in the covariates $X$ would bring faster convergence rate. I haven't thought about different parametrizations. Now I think different parametrizations in this setting don't connect very well to each other, and can't provide useful implications. 

\section{Partial Centering for Improving Convergence}\label{sec:partial-centering}

\subsection{Partial-centering for  two-level models}
Partial centering provides a continuous trade-off between centering and non-centering. %and can affect computational efficiencies significantly. 
With these parametrizations (e.g., centering, non-centering, partial centering) sharing almost the same mathematical formulation, can we derive the most efficient algorithm by optimizing over various parametrizations including not only parametrizations covered by Z\&R, but also those dictating partial centering?

Inspired by an example  in \cite{liu1999parameter} to demonstrate the power of \textit{parameter expansion},  \cite{papaspiliopoulos2003non} proposed the following \textit{partial centering parametrization} in  by introducing a constant $0\leq A\leq1$:

\vspace{2mm}
\noindent
{\bf Model S2} (Symmetric two-level model with \textbf{partial centering} parametrization). \textit{Suppose
\begin{equation}\label{partial centering 2}
    y_{ij}\sim\mathcal{N}((1-A)\mu+a_i,\sigma^2_e),\quad a_i\sim\mathcal{N}(A\mu,\sigma^2_a), \ \ i\in[1:I], j\in[1:J].
\end{equation}
where $y_{ij},\mu, a_i %\epsilon_{ij}\
in\mathbb{R}$. Same as before, a flat prior is imposed on $\mu$.}
\vspace{2mm}

A similar standard Gibbs sampler as GS(0) and GS(1) can be easily implemented. With $A=0$, \eqref{partial centering 2} reduces to non-centering parametrization; whereas with $A=1$, \eqref{partial centering 2} reduces to centering parametrization. For a general $A$, \cite{papaspiliopoulos2003non} also offered the convergence rate of the standard Gibbs sampler as
\begin{equation}
\rho_A=\frac{\left(A\sigma_a^{-2}-(1-A)J\sigma_e^{-2}\right)^2}{\left(\sigma_a^{-2}+J\sigma_e^{-2}\right)\left(A^2\sigma_a^{-2}+(1-A)^2\sigma_e^{-2}\right)}.
\end{equation}
One surprising fact is that $\rho_{A^\ast}=0$ for $A^\ast=J\sigma_e^{-2}/(\sigma_a^{-2}+J\sigma_e^{-2})$, implying that we achieve exact sampling in one step via this optimal partial centering parameterization. Note that this $A^\ast$ also results in the fact that $\mu$ and $\bar{a}$ are independent {\it a posteriori}.

\subsection{Convergence rates for three-level models}
It is of great interest to extend this flexible parametrization scheme to other models. We here provide an illustration via a slightly more complex model. 

\vspace{2mm}
\noindent
{\bf Model S3} (Symmetric three-level model with \textbf{partial centering} parametrization). \textit{With constants $A,B,C\in\mathbb{R}$, suppose
\begin{equation}\label{partial centering 3}
    y_{ijk}\sim\mathcal{N}((1-A-C)\mu+(1-B)a_i+b_{ij},\sigma^2_e),\quad b_{ij}\sim\mathcal{N}(Ba_i+C\mu,\sigma_b^2),\quad a_i\sim\mathcal{N}(A\mu,\sigma^2_a).
\end{equation}
where $y_{ij},\mu,a_i,\epsilon_{ij}\in\mathbb{R}$ and $i,j,k$ range from $1$ to $I,J,K$ respectively. Same as before, a flat prior is imposed on $\mu$.}

\vspace{2mm}
\noindent
{\bf Sampler GS($A,B,C$)} \textit{Initialize $\mu(0)$, $\ba(0)$, $\bb(0)$ and then iterate\\
1. Sample $\mu(s+1)$ from $p(\mu\mid \ba(s),\bb(s),\by)$;\\
2. Sample $a_i(s+1)$ from $p(a_i\mid \mu(s+1),\bb(s),\by)$ for all $i$;\\
3. Sample $b_{ij}(s+1)$ from $p(b_{ij}\mid\mu(s+1),\ba(s+1),\by)$ for all $i,j$.
}
\vspace{2mm}

If we select $(A,B,C)$ from $\{0,1\}^2\times\{0\}$, \eqref{partial centering 3} reduces to the four parametrizations considered in Sections 2 and 3 of Z\&R, respectively. Defining hierarchical models as trees,  Section 7 of Z\&R develop an abstract theory to deal with various parametrizations including the partial ones here, but they do not  provide more insights for cases $(A,B,C)\notin\{0,1\}^2\times\{0\}$. Let  $\tau_a=I\sigma_a^{-2},\tau_b=IJ\sigma_b^{-2},\tau_e=IJK\sigma_e^{-2}$ be the rescaled precisions. We have the following result.

\begin{Theorem}
If $(\tau_b+\tau_e)^2\tau_a+\tau_b\tau_e(\tau_b-\tau_e)\neq0$, the prescribed Gibbs sampler can achieve exact sampling in one step via suitable scalings of $A,B,C$.
\end{Theorem}

\begin{proof}
First, we  define $\delta\bbeta=\left(\delta^{(0)}\bbeta,\delta^{(1)}\bbeta,\delta^{(2)}\bbeta\right)$ exactly the same as equation (3.1) in Z\&R, where $\delta^{(0)}\bbeta=\left(\mu,\bar{a},\bar{b}\right),\bar{a}=\sum_{i}a_i/I,\bar{b}=\sum_{ij}b_{ij}/IJ$. Apply Theorem 9 in Z\&R to conclude that $\{\delta^{(0)}\bbeta\},\{\delta^{(1)}\bbeta\},\{\delta^{(2)}\bbeta\}$ evolve independently for the prescribed Gibbs sampler. 

Then, applying Theorem 11 of Z\&R, we derive the following ordering
\begin{equation*}
    \rho_{(A,B,C)}=\rho\left(\delta^{(0)}\bbeta\right)\geq\rho\left(\delta^{(1)}\bbeta\right)\geq\rho\left(\delta^{(2)}\bbeta\right)=0.
\end{equation*}
At last, we have to deal with the posterior distribution of $\delta^{(0)}\bbeta$, which is a $3$-dim Gaussian. The evolution of $\{\delta^{(0)}\bbeta(t)\}$ is simply characterized by a systematic scan Gibbs sampler, scanning according to $\mu\rightarrow\bar{a}\rightarrow\bar{b}\rightarrow\mu$. By \cite{liu1995covariance}, to obtain the convergence rate of a systematic scan Gibbs sampler, it suffices to know about pairwise correlations
\begin{align*}
    r_1&=\text{corr}(\mu,\bar{a})=\frac{BC\tau_b+A\tau_a-(1-A-C)(1-B)\tau_e}{\sqrt{C^2\tau_b+A^2\tau_a+(1-A-C)^2\tau_e}\sqrt{B^2\tau_b+\tau_a+(1-B)^2\tau_e}},\\
    r_2&=\text{corr}(\mu,\bar{b})=\frac{C\tau_b-(1-A-C)\tau_e}{\sqrt{C^2\tau_b+A^2\tau_a+(1-A-C)^2\tau_e}\sqrt{\tau_b+\tau_e}},\\
    r_3&=\text{corr}(\bar{a},\bar{b})=\frac{B\tau_b-(1-B)\tau_e}{\sqrt{\tau_b+\tau_e}\sqrt{B^2\tau_b+\tau_a+(1-B)^2\tau_e}}.
\end{align*}
By  \cite{liu1995covariance} and \cite{roberts1997updating}, we find that $\rho_{(A^\ast,B^\ast,C^\ast)}=0$ for
\begin{equation*}
    A^\ast=\frac{\tau_b\tau_e(\tau_b-\tau_e)}{(\tau_b+\tau_e)^2\tau_a+\tau_b\tau_e(\tau_b-\tau_e)}, \ 
    B^\ast=\frac{\tau_e}{\tau_b+\tau_e}, \ C^\ast=\frac{\tau_a\tau_e(\tau_b+\tau_e)}{(\tau_b+\tau_e)^2\tau_a+\tau_b\tau_e(\tau_b-\tau_e)},
\end{equation*}
due to vanishing correlations $r_1=r_2=r_3=0$.
\end{proof}

An analytical formula is available for the convergence rate of the standard Gibbs sampler GS($A,B,C$) even for general $A,B,C$. But this general formula is a little complicated and out of the scope of this article. We believe that this formula may help us understand the experimental phase transitions depicted in Figure 4 of Z\&R, and further enhance our understanding towards different parametrizations. A direct question is whether exact sampling in one step is possible for less symmetric $2,3$-level hierarchical models.

We end this section by raising more questions. Does the partial centering trick generalize to more complex structures with more confounding factors and deeper hierarchies? How do we develop partial centering for vector hierarchical models discussed in Section~\ref{sec:vector-HM} to design a better Gibbs sampler? Can we go beyond Gaussian priors to perform it in other cases, like the Poisson example in section 5 of Z\&R?

\section{Concluding Remarks}

Although Z\&R's multigrid decomposition has little to do with the classical multigrid idea for both numerical PDEs and Monte Carlo simulations, their decomposition provides a key insight to the understanding of the convergence  of  Gibbs sampling for Bayesian hierarchical models. This insight naturally leads to a constructive strategy for designing better Gibbs sampling algorithms via reparametrization for such models. Our article centers on the possibilities of extending this decomposition strategy  to  more complex, yet structured, Bayesian models,
%.   In an effort to extend Z\&R's framework to accommodate more complex models 
and to include more options (e.g., parameter expansion) for algorithmic optimization.
We specifically analyzed a few concrete examples, one in each direction. Our results are both encouraging and challenge-revealing. On one hand, we have obtained some analytical expressions of the convergence rates of various Gibbs samplers, from which we may derive an optimal parameterization; on the other hand, we find that situations become much more complex and the optimal parameterization may not exist or computable in high-dimensional cases, such as vector hierarchical models and mixed effects models. In summary, we find  that the decomposition framework established by Z\&R is both elegant and practical, and that much future endeavor is warranted for exploring and exploiting their framework.  

\bibliographystyle{apalike}
\bibliography{reference}

\end{document}